\documentclass{article}
\usepackage{paper}

\usepackage{times}
\usepackage{url}
\usepackage[hidelinks]{hyperref}
\usepackage[utf8]{inputenc}
\usepackage[small]{caption}
\usepackage{graphicx}
\usepackage{amsmath}
\usepackage{amssymb}
\usepackage{amsthm}
\usepackage{booktabs}
\usepackage{algorithm}
\usepackage{algorithmic}
\usepackage{multirow}
\usepackage{colortbl}
\usepackage{tikz}
\usetikzlibrary{shapes,arrows,arrows.meta,positioning,fit,backgrounds,calc,shapes.geometric}

\definecolor{okabe-blue}{RGB}{86,180,233}
\definecolor{okabe-orange}{RGB}{230,159,0}
\definecolor{okabe-green}{RGB}{0,158,115}
\definecolor{okabe-vermillion}{RGB}{213,94,0}
\definecolor{okabe-pink}{RGB}{204,121,167}
\definecolor{okabe-skyblue}{RGB}{86,180,233}
\definecolor{okabe-deepblue}{RGB}{0,114,178}
\definecolor{computecolor}{RGB}{0,120,200}

\newtheorem{example}{Example}
\newtheorem{theorem}{Theorem}
\newtheorem{lemma}[theorem]{Lemma}

\newtheorem{corollary}[theorem]{Corollary}
\newtheorem{definition}{Definition}
\newtheorem{assumption}{Assumption}
\newtheorem{remark}{Remark}

\newcommand{\N}{\mathcal{N}}

\newcommand{\R}{\mathbb{R}}

\newcommand{\Keff}{K_{\text{eff}}}
\newcommand{\peasy}{p_{\text{easy}}}

\title{Coalition Formation in LLM Agent Networks: \\Stability Analysis and Convergence Guarantees}

\author{
Dongxin Guo\textsuperscript{1}\and
Jikun Wu\textsuperscript{2, 3}\And
Siu-Ming Yiu\textsuperscript{1}
\affiliations
\textsuperscript{1}Department of Compute Science, The University of Hong Kong, Hong Kong, China\\
\textsuperscript{2}Brain Investing Limited, Hong Kong, China
\textsuperscript{3}Stellaris AI Limited, Hong Kong, China
\emails
bettyguo@connect.hku.hk,
hk950014@connect.hku.hk,
smyiu@cs.hku.hk
}

\begin{document}

	\maketitle

	\begin{abstract}
		Large Language Model (LLM) agents are increasingly deployed in multi-agent systems requiring strategic coordination. While recent work has analyzed LLM behavior in two-player games, \emph{coalition formation}, where $n$ agents dynamically form cooperative groups, remains theoretically uncharacterized. We present the first framework grounding coalition formation in LLM agent networks in hedonic game theory with formal stability guarantees. We introduce the \emph{LLM Coalition Formation Game} (LCFG), establish sufficient conditions for Nash-stable partitions, and prove complexity results. Our analysis reveals that LLM agents exhibit bounded rationality characterized by $\epsilon$-rational preferences; we provide both deterministic existence guarantees and consistency-driven stability bounds whose predictions are consistent with empirical outcomes. Experiments with GPT-4, Claude-3, and Llama-3 across 2,400 episodes validate our framework: LLM coalitions achieve Nash stability in 73.2\% of cases under our \emph{Coalition-of-Thought} (CoalT) protocol, compared to 58.4\% under chain-of-thought and 41.8\% under standard prompting ($p < 0.001$). Our framework provides theoretical foundations for designing stable multi-agent LLM systems.
	\end{abstract}

	\section{Introduction}

	The deployment of Large Language Model (LLM) agents in multi-stakeholder environments, from collaborative research systems \cite{wu2023autogen} to automated negotiations \cite{fu2023negotiation,abdelnabi2024cooperation}, raises fundamental questions about their strategic behavior in cooperative settings. While extensive work has characterized LLM performance in two-player strategic games \cite{fan2024rational,duan2024gtbench,buscemi2025fairgame}, the analysis of \emph{coalition formation}, where $n$ agents must dynamically organize into cooperative groups, remains conspicuously absent from the literature.

	This gap is significant for three reasons. First, real-world multi-agent deployments inherently involve coalition dynamics: research teams must form around problems, negotiating agents must identify partners, and distributed systems must allocate agents to tasks. Second, coalition formation exhibits qualitatively different strategic considerations than two-player games; agents must reason about group membership, coalition stability, and collective payoffs rather than bilateral outcomes. Third, recent work explicitly identifies this gap: studies on multi-player LLM games note that ``group size fixed at 2-3 agents precludes examination of larger-group phenomena such as coalition formation'' \cite{huynh2025llmgames}, while recent surveys on game theory and LLMs identify cooperative game theory as ``underexplored'' \cite{sun2025gametheory}.

	\textbf{Why Game Theory?} One might ask whether a formal game-theoretic framework is necessary, or if simpler heuristics suffice. Our experiments demonstrate clear value: (1) Greedy matching achieves only 52.1\% stability with lower welfare than structured approaches, showing myopic optimization fails; (2) vanilla chain-of-thought \cite{wei2022chain} achieves 58.4\%, indicating that generic structured reasoning is insufficient; the game-theoretic framing in CoalT provides an additional 14.8 percentage points improvement ($p < 0.001$); (3) our theoretical bounds are consistent with observed empirical stability rates, enabling principled protocol design.

	\textbf{Why Exclusive Partitions?} We model coalitions as a partition (each agent belongs to exactly one coalition) because in practical multi-agent LLM deployments, resource constraints (including API rate limits, compute budgets, and context window sizes) prevent an agent from simultaneously contributing to multiple coalitions within the same task episode. This mirrors team formation settings where each worker is assigned to one project at a time~\cite{chalkiadakis2011computational}. In our framework, agents are the strategic decision-makers: each agent evaluates whether it would benefit from joining a different coalition and declares its preferences accordingly (Section~\ref{sec:coalt}). A ``deviation'' means an agent's preference declaration indicates it would prefer a different group assignment.

	\begin{figure*}[t]
		\centering
		\begin{tikzpicture}[
			>=Stealth,
			font=\sffamily\small,
			scale=0.95,
			transform shape
			]

			\definecolor{inputbg}{RGB}{239,246,255}
			\definecolor{inputborder}{RGB}{147,197,253}
			\definecolor{processbg}{RGB}{254,249,235}
			\definecolor{processborder}{RGB}{252,211,77}
			\definecolor{outputbg}{RGB}{236,253,245}
			\definecolor{outputborder}{RGB}{110,231,183}

			\definecolor{gptcolor}{RGB}{59,130,246}
			\definecolor{claudecolor}{RGB}{234,88,12}
			\definecolor{llamacolor}{RGB}{16,185,129}

			\definecolor{step1color}{RGB}{99,102,241}
			\definecolor{step2color}{RGB}{14,165,233}
			\definecolor{step3color}{RGB}{245,158,11}
			\definecolor{step4color}{RGB}{236,72,153}
			\definecolor{step5color}{RGB}{139,92,246}

			\definecolor{latentcolor}{RGB}{139,92,246}
			\definecolor{theoremcolor}{RGB}{75,85,99}

			\fill[inputbg, rounded corners=6pt] (-0.2, -2.3) rectangle (3.5, 3.2);
			\draw[inputborder, rounded corners=6pt, line width=1pt] (-0.2, -2.3) rectangle (3.5, 3.2);

			\node[font=\sffamily\footnotesize\bfseries, color=gptcolor!80!black] at (1.65, 2.9) {Input Context};

			\fill[white, rounded corners=4pt] (0.1, 1.85) rectangle (3.2, 2.55);
			\draw[gptcolor, rounded corners=4pt, line width=0.8pt] (0.1, 1.75) rectangle (3.2, 2.55);
			\fill[gptcolor] (0.1, 2.35) -- (0.1, 2.51) arc(180:90:4pt) -- (3.16, 2.55) arc(90:0:4pt) -- (3.2, 2.35) -- cycle;
			\node[font=\sffamily\scriptsize\bfseries, white] at (1.65, 2.45) {GPT-4 Agents ($a_1, a_2$)};
			\node[font=\sffamily\tiny, color=black!60, anchor=east] at (0.75, 2.15) {Math};
			\fill[gptcolor!25] (0.70, 2.10) rectangle (1.55, 2.20);
			\fill[gptcolor] (0.70, 2.10) rectangle ({0.70 + 0.85*0.68}, 2.20);
			\node[font=\sffamily\tiny, color=black!50, anchor=west] at (1.62, 2.15) {0.68};
			\node[font=\sffamily\tiny, color=black!60, anchor=east] at (0.78, 1.95) {Facts};
			\fill[gptcolor!25] (0.70, 1.90) rectangle (1.55, 2.00);
			\fill[gptcolor] (0.70, 1.90) rectangle ({0.70 + 0.85*0.73}, 2.00);
			\node[font=\sffamily\tiny, color=black!50, anchor=west] at (1.62, 1.95) {0.73};

			\fill[white, rounded corners=4pt] (0.1, 0.65) rectangle (3.2, 1.55);
			\draw[claudecolor, rounded corners=4pt, line width=0.8pt] (0.1, 0.65) rectangle (3.2, 1.55);
			\fill[claudecolor] (0.1, 1.35) -- (0.1, 1.51) arc(180:90:4pt) -- (3.16, 1.55) arc(90:0:4pt) -- (3.2, 1.35) -- cycle;
			\node[font=\sffamily\scriptsize\bfseries, white] at (1.65, 1.45) {Claude-3 Agents ($a_3, a_4$)};
			\node[font=\sffamily\tiny, color=black!60, anchor=east] at (0.79, 1.15) {Facts};
			\fill[claudecolor!25] (0.70, 1.10) rectangle (1.55, 1.20);
			\fill[claudecolor] (0.70, 1.10) rectangle ({0.70 + 0.85*0.78}, 1.20);
			\node[font=\sffamily\tiny, color=black!50, anchor=west] at (1.62, 1.15) {0.78};
			\node[font=\sffamily\tiny, color=black!60, anchor=east] at (0.79, 0.95) {Logic};
			\fill[claudecolor!25] (0.70, 0.90) rectangle (1.55, 1.00);
			\fill[claudecolor] (0.70, 0.90) rectangle ({0.70 + 0.85*0.74}, 1.00);
			\node[font=\sffamily\tiny, color=black!50, anchor=west] at (1.62, 0.95) {0.74};

			\fill[white, rounded corners=4pt] (0.1, -0.55) rectangle (3.2, 0.35);
			\draw[llamacolor, rounded corners=4pt, line width=0.8pt] (0.1, -0.55) rectangle (3.2, 0.35);
			\fill[llamacolor] (0.1, 0.15) -- (0.1, 0.31) arc(180:90:4pt) -- (3.16, 0.35) arc(90:0:4pt) -- (3.2, 0.15) -- cycle;
			\node[font=\sffamily\scriptsize\bfseries, white] at (1.65, 0.25) {Llama-3 Agents ($a_5, a_6$)};
			\node[font=\sffamily\tiny, color=black!60, anchor=east] at (0.79, -0.05) {Logic};
			\fill[llamacolor!25] (0.70, -0.10) rectangle (1.55, 0.00);
			\fill[llamacolor] (0.70, -0.10) rectangle ({0.70 + 0.85*0.79}, 0.00);
			\node[font=\sffamily\tiny, color=black!50, anchor=west] at (1.62, -0.05) {0.79};
			\node[font=\sffamily\tiny, color=black!60, anchor=east] at (0.75, -0.25) {Math};
			\fill[llamacolor!25] (0.70, -0.30) rectangle (1.55, -0.20);
			\fill[llamacolor] (0.70, -0.30) rectangle ({0.70 + 0.85*0.58}, -0.20);
			\node[font=\sffamily\tiny, color=black!50, anchor=west] at (1.62, -0.25) {0.58};

			\fill[theoremcolor!8, rounded corners=3pt] (0.1, -2.0) rectangle (3.2, -0.85);
			\draw[theoremcolor!40, rounded corners=3pt, line width=0.6pt] (0.1, -2.0) rectangle (3.2, -0.85);
			\node[font=\sffamily\tiny\bfseries, color=theoremcolor!70] at (1.65, -1.05) {Coalition Formation Task};
			\node[font=\sffamily\tiny, color=theoremcolor!60, text width=2.8cm, align=center] at (1.65, -1.55) {Agents must partition into stable, complementary coalitions};

			\draw[->, line width=1.5pt, color=black!30] (3.7, 0.5) -- (4.5, 0.5);

			\fill[processbg, rounded corners=6pt] (4.7, -2.3) rectangle (11.3, 3.2);
			\draw[processborder, rounded corners=6pt, line width=1pt, dashed] (4.7, -2.3) rectangle (11.3, 3.2);

			\node[font=\sffamily\footnotesize\bfseries, color=step3color!80!black] at (8.0, 2.9) {Coalition-of-Thought (CoalT)};

			\fill[white, rounded corners=3pt] (5.0, 1.6) rectangle (5.8, 2.55);
			\draw[step1color, rounded corners=3pt, line width=0.7pt] (5.0, 1.6) rectangle (5.8, 2.55);
			\fill[step1color] (5.0, 2.35) -- (5.0, 2.52) arc(180:90:3pt) -- (5.77, 2.55) arc(90:0:3pt) -- (5.8, 2.35) -- cycle;
			\node[font=\sffamily\fontsize{6}{7}\selectfont\bfseries, white] at (5.4, 2.45) {Step 1};
			\node[font=\sffamily\fontsize{4.5}{6.5}\selectfont, color=black!60, text width=0.7cm, align=center] at (5.4, 1.95) {Capability\\Analysis};

			\fill[white, rounded corners=3pt] (6.15, 1.6) rectangle (6.95, 2.55);
			\draw[step2color, rounded corners=3pt, line width=0.7pt] (6.15, 1.6) rectangle (6.95, 2.55);
			\fill[step2color] (6.15, 2.35) -- (6.15, 2.52) arc(180:90:3pt) -- (6.92, 2.55) arc(90:0:3pt) -- (6.95, 2.35) -- cycle;
			\node[font=\sffamily\fontsize{6}{7}\selectfont\bfseries, white] at (6.55, 2.45) {Step 2};
			\node[font=\sffamily\fontsize{4.5}{6.5}\selectfont, color=black!60, text width=0.7cm, align=center] at (6.55, 1.95) {Complem.\\Assess.};

			\fill[white, rounded corners=3pt] (7.3, 1.6) rectangle (8.1, 2.55);
			\draw[step3color, rounded corners=3pt, line width=0.7pt] (7.3, 1.6) rectangle (8.1, 2.55);
			\fill[step3color] (7.3, 2.35) -- (7.3, 2.52) arc(180:90:3pt) -- (8.07, 2.55) arc(90:0:3pt) -- (8.1, 2.35) -- cycle;
			\node[font=\sffamily\fontsize{6}{7}\selectfont\bfseries, white] at (7.7, 2.45) {Step 3};
			\node[font=\sffamily\fontsize{4.5}{6.5}\selectfont, color=black!60, text width=0.7cm, align=center] at (7.68, 1.95) {Value\\Estimation};

			\fill[white, rounded corners=3pt] (8.45, 1.6) rectangle (9.25, 2.55);
			\draw[step4color, rounded corners=3pt, line width=0.7pt] (8.45, 1.6) rectangle (9.25, 2.55);
			\fill[step4color] (8.45, 2.35) -- (8.45, 2.52) arc(180:90:3pt) -- (9.22, 2.55) arc(90:0:3pt) -- (9.25, 2.35) -- cycle;
			\node[font=\sffamily\fontsize{6}{7}\selectfont\bfseries, white] at (8.85, 2.45) {Step 4};
			\node[font=\sffamily\fontsize{4.5}{6.5}\selectfont, color=black!60, text width=0.7cm, align=center] at (8.85, 1.95) {Coord.\\Cost};

			\fill[white, rounded corners=3pt] (9.6, 1.6) rectangle (10.4, 2.55);
			\draw[step5color, rounded corners=3pt, line width=0.7pt] (9.6, 1.6) rectangle (10.4, 2.55);
			\fill[step5color] (9.6, 2.35) -- (9.6, 2.52) arc(180:90:3pt) -- (10.37, 2.55) arc(90:0:3pt) -- (10.4, 2.35) -- cycle;
			\node[font=\sffamily\fontsize{6}{7}\selectfont\bfseries, white] at (10.0, 2.45) {Step 5};
			\node[font=\sffamily\fontsize{4.5}{6.5}\selectfont, color=black!60, text width=0.7cm, align=center] at (9.96, 1.95) {Preference\\Declare};

			\draw[-{Stealth[length=2pt, width=2pt]}, line width=0.5pt, color=black!50] (5.85, 2.05) -- (6.1, 2.05);
			\draw[-{Stealth[length=2pt, width=2pt]}, line width=0.5pt, color=black!50] (7.0, 2.05) -- (7.25, 2.05);
			\draw[-{Stealth[length=2pt, width=2pt]}, line width=0.5pt, color=black!50] (8.15, 2.05) -- (8.4, 2.05);
			\draw[-{Stealth[length=2pt, width=2pt]}, line width=0.5pt, color=black!50] (9.3, 2.05) -- (9.55, 2.05);

			\node[font=\sffamily\scriptsize\bfseries, color=latentcolor!70] at (8.0, 1.25) {Coalition Formation in Latent Space};

			\fill[latentcolor!5, rounded corners=4pt] (5.0, -1.8) rectangle (7.0, 0.95);
			\draw[latentcolor!25, rounded corners=4pt, line width=0.6pt] (5.0, -1.8) rectangle (7.0, 0.95);
			\node[font=\sffamily\tiny\bfseries, color=latentcolor!50] at (6.0, 0.75) {Before};

			\fill[gptcolor] (5.4, 0.35) circle (4pt);
			\fill[gptcolor] (6.4, -0.1) circle (4pt);
			\fill[claudecolor] (5.7, -0.45) circle (4pt);
			\fill[claudecolor] (6.55, 0.35) circle (4pt);
			\fill[llamacolor] (5.35, -0.9) circle (4pt);
			\fill[llamacolor] (6.2, -1.35) circle (4pt);

			\node[font=\sffamily\tiny, color=black!40] at (5.4, 0.6) {$a_1$};
			\node[font=\sffamily\tiny, color=black!40] at (6.4, 0.15) {$a_2$};
			\node[font=\sffamily\tiny, color=black!40] at (5.7, -0.2) {$a_3$};
			\node[font=\sffamily\tiny, color=black!40] at (6.55, 0.6) {$a_4$};
			\node[font=\sffamily\tiny, color=black!40] at (5.35, -0.65) {$a_5$};
			\node[font=\sffamily\tiny, color=black!40] at (6.2, -1.1) {$a_6$};

			\draw[->, line width=1.2pt, color=latentcolor!50] (7.25, -0.4) -- (8.75, -0.4);
			\node[font=\sffamily\tiny, color=latentcolor!60, fill=white, inner sep=2pt] at (8.0, -0.15) {$h \leftarrow h + \alpha \cdot \Delta h$};
			\node[font=\sffamily\tiny, color=latentcolor!40] at (8.0, -0.7) {Iterative $(T=k)$};

			\fill[latentcolor!5, rounded corners=4pt] (9.0, -1.8) rectangle (11.0, 0.95);
			\draw[latentcolor!25, rounded corners=4pt, line width=0.6pt] (9.0, -1.8) rectangle (11.0, 0.95);
			\node[font=\sffamily\tiny\bfseries, color=latentcolor!50] at (10.0, 0.75) {After};

			\draw[outputborder, rounded corners=4pt, line width=1pt, fill=outputbg] (9.25, 0.15) rectangle (10.05, 0.6);
			\fill[gptcolor] (9.45, 0.38) circle (4pt);
			\fill[claudecolor] (9.85, 0.38) circle (4pt);
			\node[font=\sffamily\tiny, color=theoremcolor!60] at (9.65, 0.0) {$C_1$};

			\draw[outputborder, rounded corners=4pt, line width=1pt, fill=outputbg] (9.2, -0.75) rectangle (10.55, -0.15);
			\fill[gptcolor] (9.45, -0.45) circle (4pt);
			\fill[llamacolor] (9.85, -0.45) circle (4pt);
			\fill[llamacolor] (10.25, -0.45) circle (4pt);
			\node[font=\sffamily\tiny, color=theoremcolor!60] at (9.85, -0.9) {$C_2$};

			\draw[outputborder, rounded corners=4pt, line width=1pt, fill=outputbg] (9.75, -1.5) rectangle (10.35, -1.1);
			\fill[claudecolor] (10.05, -1.3) circle (4pt);
			\node[font=\sffamily\tiny, color=theoremcolor!60] at (10.05, -1.65) {$C_3$};

			\draw[->, line width=1.5pt, color=black!30] (11.5, 0.5) -- (12.3, 0.5);

			\fill[outputbg, rounded corners=6pt] (12.5, -2.3) rectangle (16.2, 3.2);
			\draw[outputborder, rounded corners=6pt, line width=1pt] (12.5, -2.3) rectangle (16.2, 3.2);

			\node[font=\sffamily\footnotesize\bfseries, color=llamacolor!80!black] at (14.35, 2.9) {Output};

			\node[font=\sffamily\scriptsize\bfseries, color=theoremcolor!70] at (14.35, 2.5) {Nash-Stable Partition $\pi^*$};

			\fill[white, rounded corners=3pt] (12.8, 1.5) rectangle (15.9, 2.15);
			\draw[outputborder, rounded corners=3pt, line width=0.8pt] (12.8, 1.5) rectangle (15.9, 2.15);
			\node[font=\sffamily\tiny\bfseries, color=theoremcolor!60] at (13.1, 1.82) {$C_1$:};
			\fill[gptcolor] (13.5, 1.82) circle (5pt);
			\node[font=\sffamily\tiny, white] at (13.5, 1.82) {$a_1$};
			\fill[claudecolor] (14.1, 1.82) circle (5pt);
			\node[font=\sffamily\tiny, white] at (14.1, 1.82) {$a_3$};
			\node[font=\sffamily\tiny, color=llamacolor!80!black] at (15.4, 1.82) {\checkmark Stable};

			\fill[white, rounded corners=3pt] (12.8, 0.65) rectangle (15.9, 1.3);
			\draw[outputborder, rounded corners=3pt, line width=0.8pt] (12.8, 0.65) rectangle (15.9, 1.3);
			\node[font=\sffamily\tiny\bfseries, color=theoremcolor!60] at (13.1, 0.97) {$C_2$:};
			\fill[gptcolor] (13.5, 0.97) circle (5pt);
			\node[font=\sffamily\tiny, white] at (13.5, 0.97) {$a_2$};
			\fill[llamacolor] (14.1, 0.97) circle (5pt);
			\node[font=\sffamily\tiny, white] at (14.1, 0.97) {$a_5$};
			\fill[llamacolor] (14.7, 0.97) circle (5pt);
			\node[font=\sffamily\tiny, white] at (14.7, 0.97) {$a_6$};
			\node[font=\sffamily\tiny, color=llamacolor!80!black] at (15.4, 0.97) {\checkmark Stable};

			\fill[white, rounded corners=3pt] (12.8, -0.2) rectangle (15.9, 0.45);
			\draw[outputborder, rounded corners=3pt, line width=0.8pt] (12.8, -0.2) rectangle (15.9, 0.45);
			\node[font=\sffamily\tiny\bfseries, color=theoremcolor!60] at (13.1, 0.12) {$C_3$:};
			\fill[claudecolor] (13.5, 0.12) circle (5pt);
			\node[font=\sffamily\tiny, white] at (13.5, 0.12) {$a_4$};
			\node[font=\sffamily\tiny, color=llamacolor!80!black] at (15.4, 0.12) {\checkmark Stable};

			\fill[outputborder!20, rounded corners=6pt] (12.8, -1.35) rectangle (15.9, -0.5);
			\draw[outputborder, rounded corners=6pt, line width=1pt] (12.8, -1.35) rectangle (15.9, -0.5);
			\node[font=\sffamily\Large\bfseries, color=llamacolor!85!black] at (14.35, -0.75) {73.2\%};
			\node[font=\sffamily\tiny, color=theoremcolor!60] at (14.35, -1.1) {Nash Stability Rate};

			\fill[white, rounded corners=3pt] (12.8, -2.15) rectangle (15.9, -1.55);
			\draw[theoremcolor!20, rounded corners=3pt, line width=0.5pt] (12.8, -2.15) rectangle (15.9, -1.55);
			\node[font=\sffamily\tiny, color=theoremcolor!50] at (13.4, -1.72) {Standard:};
			\node[font=\sffamily\tiny, color=theoremcolor!50] at (14.2, -1.72) {41.8\%};
			\node[font=\sffamily\tiny, color=theoremcolor!50] at (15.15, -1.72) {CoT: 58.4\%};
			\node[font=\sffamily\tiny\bfseries, color=llamacolor!80!black] at (14.4, -1.97) {CoalT: 73.2\%};

			\fill[theoremcolor!5, rounded corners=5pt] (4.3, -3.4) rectangle (12.0, -2.6);
			\draw[theoremcolor!30, rounded corners=5pt, line width=0.7pt] (4.3, -3.4) rectangle (12.0, -2.6);
			\node[font=\sffamily\footnotesize\bfseries, color=theoremcolor!70] at (5.35, -3.0) {Theorem 2:};
			\node[font=\small, color=theoremcolor!60] at (7.85, -3.0) {$\Pr[\text{Nash-stable}] \propto p^{K_{\text{eff}}}$};
			\node[font=\sffamily\tiny, color=theoremcolor!50] at (10.75, -3.0) {\textit{Consistency} $\Rightarrow$ \textit{Stability}};

			\fill[white, rounded corners=3pt] (-0.2, -3.4) rectangle (4.0, -2.6);
			\draw[theoremcolor!15, rounded corners=3pt, line width=0.5pt] (-0.2, -3.4) rectangle (4.0, -2.6);
			\fill[gptcolor] (0.2, -3.0) circle (4pt);
			\node[font=\sffamily\tiny, color=theoremcolor!60] at (0.8, -3.0) {GPT-4};
			\fill[claudecolor] (1.5, -3.0) circle (4pt);
			\node[font=\sffamily\tiny, color=theoremcolor!60] at (2.2, -3.0) {Claude-3};
			\fill[llamacolor] (3.0, -3.0) circle (4pt);
			\node[font=\sffamily\tiny, color=theoremcolor!60] at (3.6, -3.0) {Llama-3};

		\end{tikzpicture}
		\caption{\textbf{LCFG Framework Overview.} \textit{Left}: Input context with heterogeneous LLM agents showing capability profiles across Math, Facts, and Logic dimensions. \textit{Center}: Coalition-of-Thought (CoalT) reasoning module, a 5-step pipeline that guides agents through structured coalition evaluation. The latent space visualization shows agents transitioning from scattered initial positions to clustered Nash-stable coalitions through iterative refinement ($h \leftarrow h + \alpha \cdot \Delta h$). The dashed border indicates the trainable/promptable module. \textit{Right}: Output Nash-stable partition $\pi^*$ achieving 73.2\% stability rate (+14.8pp over vanilla CoT). \textbf{Key insight}: Theorem~\ref{thm:probabilistic} shows stability scales with preference consistency ($p^{K_{\text{eff}}}$), not perfect rationality.}
		\label{fig:framework}
	\end{figure*}
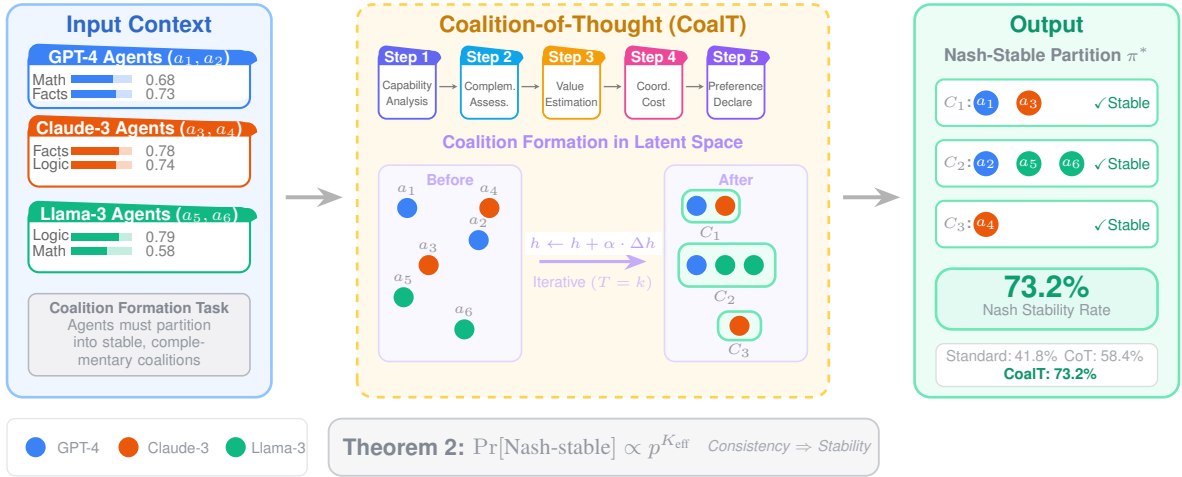

	\textbf{Contributions.} We address this gap through four contributions:
	\begin{enumerate}
		\item \textbf{Formal Framework}: We introduce the \emph{LLM Coalition Formation Game} (LCFG), extending hedonic game theory \cite{bogomolnaia2002stability} to characterize LLM agents with capability profiles and bounded-rational preferences (Section~\ref{sec:framework}).

		\item \textbf{Stability Theorems}: We establish \emph{two complementary guarantees}: (i) deterministic existence under ideal conditions ($\epsilon < \delta/2$), and (ii) \emph{consistency-driven stability bounds} using logit dynamics analysis that decompose stability into preference consistency and game structure factors, providing a conservative lower bound consistent with observed outcomes (Section~\ref{sec:theory}).

		\item \textbf{Complexity Analysis}: We show that \emph{computing} a Nash-stable partition is NP-hard in general hedonic games, motivating structural assumptions under which LCFGs become polynomial-time solvable via \emph{capability-monotonicity} constraints that hold in practice (Section~\ref{sec:complexity}).

		\item \textbf{Experimental Validation}: Through 2,400 coalition formation episodes across three LLM architectures with \emph{five baselines including vanilla CoT}, we validate our theoretical framework with formal statistical testing and introduce \emph{Coalition-of-Thought} (CoalT), a prompting protocol achieving significant stability improvements (Section~\ref{sec:experiments}).
	\end{enumerate}

	\section{Related Work}
	\label{sec:related}

	\textbf{Coalition Game Theory.} Coalition formation in multi-agent systems has been extensively studied in classical settings \cite{chalkiadakis2011computational,rahwan2015coalition}. Hedonic games \cite{dreze1980hedonic,bogomolnaia2002stability} model agents with preferences over coalition membership, with well-characterized stability concepts: core stability, Nash stability, and individual stability form a hierarchy of solution concepts \cite{aziz2019hedonic}. Recent work establishes complexity bounds for random hedonic games \cite{bullinger2024random} and addresses popularity in fractional hedonic games \cite{bullinger2025settling}, but none considers LLM agents. Our work bridges this gap by extending hedonic game theory to account for the bounded rationality and stochastic preferences inherent in LLM decision-making.

	\textbf{LLMs in Strategic Settings.} GTBench \cite{duan2024gtbench} and FAIRGAME \cite{buscemi2025fairgame} provide frameworks for analyzing LLM strategic behavior in competitive games, revealing that LLM performance varies significantly across game types and reasoning complexity. LLMArena \cite{chen2024llmarena} and Alympics \cite{mao2025alympics} further demonstrate LLM capabilities in dynamic multi-agent environments. LLM-Deliberation \cite{abdelnabi2024cooperation} extends this to multi-party negotiation scenarios but without formal coalition-theoretic analysis. Recent work on non-cooperative equilibria includes ECON \cite{jiayang2025econ}, which proves existence of Bayesian Nash Equilibrium in multi-LLM coordination, and studies showing that LLM performance degrades with strategic complexity \cite{gandhi2023strategic}. Critically, these works focus on two-player or fixed-group competitive/mixed-motive games; extension to $n$-player coalitional settings with dynamic group formation requires fundamentally new theoretical machinery, which we provide.

	\textbf{Cooperative Game Theory for LLMs.} A growing body of work applies cooperative game theory to LLM agents, primarily through Shapley value-based credit assignment. Shapley-Coop \cite{hua2025shapleycoop} introduces ``Shapley Chain-of-Thought'' for LLM agents to estimate marginal contributions within multi-agent teams, and DAO-Agent \cite{daoagent2025} defines explicit characteristic functions over LLM agent coalitions for blockchain-verified reward allocation. However, these approaches assume a fixed grand coalition and compute contribution scores within it; they do not address the \emph{partition optimization} problem of which coalitions should form, nor do they analyze stability of the resulting structures. Separately, Briman et al.~\cite{briman2025coalition} model LLM-mediated coalition formation where agents have Euclidean preferences and an LLM generates textual compromise proposals. While this work features formal coalition structures and iterative partition evolution, LLMs serve as mediator tools rather than autonomous strategic players, and stability analysis is limited to convergence properties without formal guarantees (Nash, core, or individual stability). Kulkarni et al.~\cite{kulkarni2025dynamic} use LLMs with hypergame theory to \emph{detect} existing coalitions in Diplomacy from natural language, rather than forming them. Our framework differs from all of these in three respects: LLM agents are the strategic players (not tools), the objective is partition optimization with stability guarantees, and we provide a dedicated prompting protocol (CoalT) for coalition reasoning.

	\textbf{Multi-Agent LLM Systems.} AutoGen \cite{wu2023autogen}, CAMEL \cite{li2023camel}, and MetaGPT \cite{hong2023metagpt} enable multi-LLM coordination for complex tasks. AgentVerse \cite{chen2024agentverse} observes emergent behaviors including spontaneous cooperation and role specialization, while ChatDev \cite{qian2024chatdev} demonstrates effective software development through LLM collaboration. However, these systems lack game-theoretic foundations for understanding \emph{when} and \emph{why} stable coordination emerges. Our framework fills this gap by providing formal conditions under which LLM agents converge to stable coalitions, enabling principled system design rather than trial-and-error tuning.

	\textbf{Bounded Rationality and QRE.} Our $\epsilon$-rationality model connects to classical behavioral game theory \cite{simon1955behavioral,camerer2003behavioral}. Quantal response equilibria (QRE) \cite{mckelvey1995quantal} model agents who make errors proportional to payoff differences: higher-value options are more likely chosen but not deterministically. We formally establish the connection: under logit dynamics with parameter $\lambda = 1/\epsilon$, our model yields QRE-equivalent choice probabilities (Appendix~\ref{app:qre}). This connection enables us to leverage decades of QRE analysis while adapting the framework to LLM-specific characteristics such as prompt-dependent consistency.

	\textbf{Positioning.} Table~\ref{tab:related} summarizes how our work relates to prior approaches. While recent work applies Shapley values to LLM agent credit assignment \cite{hua2025shapleycoop} and uses LLMs as mediators for coalition negotiation \cite{briman2025coalition}, we are, to the best of our knowledge, the first to combine: (i) a formal coalitional game model where LLM agents are strategic players, (ii) partition-based stability analysis (Nash stability, core stability) with provable guarantees, and (iii) a prompting protocol specifically designed for coalition reasoning.

	\begin{table}[t]
		\centering
		\caption{Comparison with related approaches. \emph{Coalitions}: models coalition partition structures. \emph{Stability}: formal game-theoretic stability analysis. \emph{Protocol}: dedicated reasoning protocol for coalition decisions.}
		\label{tab:related}
		\small
		\begin{tabular}{@{}lcccc@{}}
			\toprule
			\textbf{Work} & \textbf{$n$-player} & \textbf{Coalitions} & \textbf{Stability} & \textbf{Protocol} \\
			\midrule
			GTBench & \texttimes & \texttimes & \texttimes & \texttimes \\
			FAIRGAME & \texttimes & \texttimes & \texttimes & \texttimes \\
			Shapley-Coop & \checkmark & \texttimes & \texttimes & \texttimes \\
			Briman et al. & \checkmark & \checkmark & \texttimes & \texttimes \\
			Kulkarni et al. & \checkmark & \checkmark$^*$ & \texttimes & \texttimes \\
			AutoGen & \checkmark & \texttimes & \texttimes & \texttimes \\
			\textbf{Ours} & \checkmark & \checkmark & \checkmark & \checkmark \\
			\bottomrule
			\multicolumn{5}{l}{\scriptsize $^*$Detection only, not formation.} \\
		\end{tabular}
	\end{table}

	\section{Problem Formulation}
	\label{sec:framework}

	We formalize coalition formation among LLM agents as a \emph{hedonic game}~\cite{dreze1980hedonic,bogomolnaia2002stability}, a class of cooperative games where each agent's preferences depend solely on the members of its own coalition. This is natural for LLM multi-agent systems: an agent's performance in a coalition depends on who it collaborates with, not on how other coalitions are organized.

	\subsection{LLM Coalition Formation Games}

	\begin{definition}[LLM Agent]
		\label{def:agent}
		An \emph{LLM agent} is a tuple $a_i = (m_i, \theta_i, \mathbf{c}_i)$ where $m_i \in \mathcal{M}$ is the model architecture (e.g., GPT-4, Claude-3), $\theta_i \in \Theta = [0,2] \times \Sigma^*$ specifies configuration (temperature $\tau \in [0,2]$ and system prompt $s \in \Sigma^*$), and $\mathbf{c}_i \in [0,1]^d$ is a \emph{capability profile} over $d$ skill dimensions (e.g., mathematical reasoning, factual knowledge, logical analysis).
	\end{definition}

	Capability profiles are estimated empirically by evaluating each LLM on domain-specific benchmarks (Section~\ref{sec:experiments}). They serve as the bridge between the abstract hedonic game model and concrete LLM agents: the coalition value function is defined in terms of the capability profiles of its members.

	\begin{definition}[LLM Coalition Formation Game]
		\label{def:lcfg}
		An \emph{LLM Coalition Formation Game} (LCFG) is a hedonic game $G = (\N, v, \{\succsim_i\}_{i \in \N})$ where:
		\begin{itemize}
			\item $\N = \{a_1, \ldots, a_n\}$ is a set of LLM agents (Definition~\ref{def:agent})
			\item $v: 2^\N \to \R$ is a \emph{coalition value function} parameterized by member capability profiles
			\item $\succsim_i$ is agent $a_i$'s \emph{preference relation} over coalitions containing $a_i$, informed by per-capita value $v_i(S) = v(S)/|S|$ (see Definition~\ref{def:epsilon})
		\end{itemize}
	\end{definition}

	The coalition value function captures joint task performance:
	\begin{equation}
		\label{eq:value}
		v(S) = \phi\left(\bigoplus_{a_i \in S} \mathbf{c}_i\right) - \psi(|S|)
	\end{equation}
	where $\phi: [0,1]^d \to \R$ aggregates capabilities (we use $\phi(\mathbf{c}) = \|\mathbf{c}\|_1/d$), $\oplus$ denotes componentwise maximum (modeling coverage-based tasks where having any expert suffices), and $\psi: \mathbb{N} \to \R$ captures coordination costs. We model $\psi(k) = \alpha \cdot k^{\beta}$ with $\alpha = 0.15$ and $\beta = 1.3$, empirically calibrated to minimize prediction error (see Appendix~\ref{app:value_function} for sensitivity analysis). The superlinear scaling ($\beta > 1$) is consistent with coordination cost findings~\cite{chalkiadakis2011computational,nisan2007algorithmic}.

	\begin{example}[Worked Example]
		\label{ex:worked}
		Consider three agents with capabilities $\mathbf{c}_1 = (0.68, 0.30, 0.40)$, $\mathbf{c}_2 = (0.40, 0.65, 0.30)$, $\mathbf{c}_3 = (0.30, 0.40, 0.76)$. Using $\phi(\mathbf{c}) = \|\mathbf{c}\|_1/d$ and $\psi(k) = 0.15k^{1.3}$: coalition $\{a_1, a_2\}$ has $v = 0.21$ (per-capita 0.10), while the grand coalition $\{a_1, a_2, a_3\}$ has $v = 0.07$ (per-capita 0.02). Despite higher capability coverage, the grand coalition yields lower per-capita value due to coordination costs, a core tradeoff our framework captures.
	\end{example}

	\subsection{Bounded-Rational Preferences}

	Classical hedonic games assume complete, transitive preferences. LLM agents exhibit \emph{bounded rationality}, i.e., systematic deviations from optimal decision-making:

	\begin{definition}[$\epsilon$-Rational Preferences]
		\label{def:epsilon}
		Agent $a_i$'s preferences are \emph{$\epsilon$-rational} with respect to value function $v$ if for all coalitions $S, T \ni a_i$:
		\begin{equation}
			v_i(S) > v_i(T) + \epsilon \implies S \succ_i T
		\end{equation}
		where $v_i(S) = v(S)/|S|$ is $a_i$'s per-capita value in coalition $S$.
	\end{definition}

	Intuitively, $\epsilon$-rationality captures that LLM agents reliably identify the better option when value differences exceed $\epsilon$, but may make inconsistent choices for closer comparisons. Our experiments (Section~\ref{sec:experiments}) estimate $\epsilon \approx 0.15$ for GPT-4 and $\epsilon \approx 0.22$ for Llama-3. These values are estimated by measuring the value gap threshold below which agent choices become near-random (see Appendix~\ref{app:qre} for methodology). Under logit choice models from Quantal Response Equilibrium (QRE) theory \cite{mckelvey1995quantal}, this corresponds to precision parameters $\lambda = 1/\epsilon \approx 6.7$ (GPT-4) and $\lambda \approx 4.5$ (Llama-3), indicating substantial but bounded deviations from perfect rationality.

	\subsection{Coalition Structures and Stability}

	A \emph{coalition structure} $\pi = \{C_1, \ldots, C_k\}$ is a partition of $\N$. We adopt standard stability concepts:

	\begin{definition}[Stability Concepts]
		\label{def:stability}
		Given LCFG $G$ and coalition structure $\pi$:
		\begin{itemize}
			\item $\pi$ is \emph{Nash-stable} if no agent unilaterally prefers joining another coalition: $\forall a_i \in C \in \pi$, $\forall C' \in \pi \cup \{\emptyset\}$: $C \succsim_i C' \cup \{a_i\}$
			\item $\pi$ is \emph{individually stable} (IS) if no agent can profitably deviate without making any member of the receiving coalition worse off
			\item $\pi$ is \emph{core-stable} if no group of agents can jointly deviate such that all members of the deviating group strictly benefit
		\end{itemize}
	\end{definition}

	Nash stability is the strongest individual-deviation concept: it requires that no agent wants to switch, regardless of whether the receiving coalition consents. Individual stability weakens this by requiring that the receiving coalition does not object. Core stability addresses group deviations: Nash stability $\Rightarrow$ Individual stability, but Nash stability and core stability are generally incomparable in hedonic games~\cite{bogomolnaia2002stability,aziz2019hedonic}. We focus on Nash stability as our primary notion because it is verifiable in polynomial time under our capability-monotonicity assumption (Theorem~\ref{thm:complexity}) and provides strong guarantees against unilateral deviations.

	\textbf{Existence of Nash-Stable Partitions.} Nash-stable partitions are not guaranteed to exist in all hedonic games. Under per-capita value splitting $v_i(S) = v(S)/|S|$ with componentwise-max aggregation, cycles can arise when a low-capability agent benefits from joining a high-capability agent's coalition while the latter prefers to be alone (see Appendix~\ref{app:counterexample} for a formal example). Our Theorem~\ref{thm:existence} identifies sufficient conditions for existence, while Theorem~\ref{thm:probabilistic} provides probabilistic guarantees when these conditions are relaxed.

	\section{Theoretical Analysis}
	\label{sec:theory}

	We establish when LLM agents converge to stable coalition structures, providing both \emph{deterministic guarantees} under ideal conditions and \emph{consistency-driven bounds} for realistic settings.

	\subsection{Deterministic Existence (Ideal Conditions)}

	\begin{assumption}[Value Gap Condition]
		\label{ass:gap}
		An LCFG $G$ satisfies the \emph{$\delta$-value gap condition} if there exists $\delta > 0$ such that for all agents $a_i$ and distinct coalitions $S, T \ni a_i$: either $v_i(S) = v_i(T)$ or $|v_i(S) - v_i(T)| \geq \delta$.
	\end{assumption}

	This assumption is naturally satisfied with benchmark-derived capability profiles. We empirically verify $\delta \approx 0.08$ through exhaustive enumeration of all coalition pairs in our 6-agent setting (Appendix~\ref{app:delta_verification}). Note that $\delta$ is a property of the \emph{specific game instance} derived from our capability profiles and value function parameterization, not a universal constant from the hedonic games literature.

	\begin{assumption}[Potential Alignment]
		\label{ass:potential}
		An LCFG $G$ satisfies \emph{potential alignment} if for the potential function $\Phi(\pi) = \sum_{C \in \pi} v(C)$, every improving deviation (where $v_i(C' \cup \{a_i\}) > v_i(C)$) also increases $\Phi$: $\Phi(\pi') > \Phi(\pi)$.
	\end{assumption}

	Potential alignment holds when agents bring sufficient unique value to receiving coalitions; specifically, when the marginal value $v(C' \cup \{a_i\}) - v(C')$ exceeds the loss $v(C) - v(C \setminus \{a_i\})$. Under componentwise-max aggregation with diverse multi-dimensional capability profiles ($d \geq 2$), this is naturally satisfied when each agent contributes distinct coverage. It may fail in degenerate cases (e.g., one-dimensional capabilities where high-capability agents gain nothing from low-capability partners; see Appendix~\ref{app:counterexample}).

	\begin{theorem}[Deterministic Existence]
		\label{thm:existence}
		Let $G$ be an LCFG satisfying the $\delta$-value gap condition (Assumption~\ref{ass:gap}), potential alignment (Assumption~\ref{ass:potential}), and \emph{capability monotonicity}: for all $S \subseteq \N$ and agents $a_i, a_j$, if $\mathbf{c}_i \leq \mathbf{c}_j$ componentwise, then $v(S \cup \{a_i\}) \leq v(S \cup \{a_j\})$. If all agents have $\epsilon$-rational preferences with $\epsilon < \delta/2$, then a Nash-stable partition exists and can be found in polynomial time.
	\end{theorem}

	\begin{proof}[Proof Sketch]
		Under $\epsilon$-rationality and the $\delta$-value gap condition, an improving deviation requires a per-capita value increase of at least $\delta - \epsilon > \delta/2 > 0$. By potential alignment, each such deviation strictly increases $\Phi(\pi) = \sum_{C \in \pi} v(C)$. Since $\Phi$ is bounded above and takes values from a discrete set (by the value gap condition), the improvement path terminates at a Nash-stable partition in at most $O(n \cdot \Delta_v / \delta)$ steps, where $\Delta_v$ is the value range. Full proof in Appendix~\ref{app:proof_existence}.
	\end{proof}

	\begin{remark}[Scope and Limitations of Theorem~\ref{thm:existence}]
		\label{rem:scope}
		The condition $\epsilon < \delta/2$ ensures that $\epsilon$-rational agents can distinguish coalitions with different values: when $\epsilon \geq \delta/2$, the agent's error margin exceeds the value gap, and rational discrimination breaks down. In practice, LLM agents exhibit $\epsilon \approx 0.15$--$0.22$, exceeding $\delta/2 \approx 0.04$. This motivates our main theoretical result (Theorem~\ref{thm:probabilistic}), which provides probabilistic guarantees without requiring $\epsilon < \delta/2$.
	\end{remark}

	\subsection{Consistency-Driven Stability (Realistic Conditions)}

	When $\epsilon > \delta/2$, deterministic guarantees no longer hold. We characterize stability using \emph{logit dynamics} from behavioral game theory \cite{mckelvey1995quantal}.

	\textbf{Key Insight.} Stability depends primarily on \emph{preference consistency}, i.e., whether agents make the same choice when queried repeatedly, rather than perfect rationality. An agent who consistently prefers suboptimal coalitions still contributes to stable outcomes, while an agent with perfect utility but inconsistent choices creates instability.

	\begin{definition}[Preference Consistency]
		\label{def:consistency}
		Agent $a_i$ has \emph{preference consistency} $p_i \in [0,1]$ if, across independent queries for the same coalition comparison, the agent returns the same preference with probability $p_i$.
	\end{definition}

	\begin{definition}[Effective Critical Decisions]
		\label{def:keff}
		Let $\Keff$ denote the number of \emph{effectively critical} decisions, those where the value gap $|\Delta v| < 2\epsilon$, making consistency meaningfully less than 1. Let $K_n - \Keff$ denote ``easy'' decisions with large value gaps where consistency approaches $\peasy \approx 1$.
	\end{definition}

	\begin{theorem}[Consistency-Driven Stability Bound]
		\label{thm:probabilistic}
		Let $G$ be an LCFG with $n$ agents satisfying capability monotonicity. Let $p \in (0,1]$ be the preference consistency on critical decisions, $\peasy \approx 1$ on easy decisions, $\Keff$ the number of critical decisions, and $K_n$ total decisions. Assuming decision independence across agents and coalition comparisons:
		\begin{equation}
			\label{eq:stability_bound}
			\Pr[\text{Nash-stable}] \geq p^{\Keff} \cdot \peasy^{K_n - \Keff} \cdot \gamma(G)
		\end{equation}
		where $\gamma(G)$ is the probability that consistent dynamics reach a Nash-stable partition. Under logit dynamics with precision $\lambda = 1/\bar{\epsilon}$, $\gamma(G)$ is bounded below by $1 - \exp(-\delta/\bar{\epsilon})$, with $\bar{\epsilon}$ the mean rationality bound and $\delta$ the value gap.
	\end{theorem}

	\begin{proof}[Proof Sketch]
		By conditional probability, $\Pr[\text{Nash-stable}] = \Pr[\text{consistent}] \cdot \Pr[\text{Nash} \mid \text{consistent}]$. The first factor, under decision independence, is at least $p^{\Keff} \cdot \peasy^{K_n - \Keff}$. The second factor is $\gamma(G)$: given consistent preferences, the system follows deterministic improving dynamics on a potential game, and the stationary distribution of logit dynamics concentrates on potential maxima~\cite{blume1993statistical}. Full derivation in Appendix~\ref{app:proof_probabilistic}.
	\end{proof}

	\textbf{Empirical Validation.} For our setting ($n=6$, $\Keff \approx 5$ critical decisions out of $K_n \approx 15$ total, $p = 0.86$, $\peasy = 0.98$), we estimate $\gamma(G) \approx 0.90$ empirically as the fraction of episodes with consistent decisions that reach Nash stability. This exceeds the formula lower bound $1 - \exp(-0.08/0.17) \approx 0.38$, indicating favorable game structure. Our theoretical lower bound is:
	\begin{equation}
		\Pr[\text{Nash-stable}] \geq 0.86^5 \cdot 0.98^{10} \cdot 0.90 \approx 0.35
	\end{equation}
	The observed rate of 73.2\% substantially exceeds this lower bound, confirming that the consistency-driven framework captures the dominant factor in stability. The monotonic relationship between $p$ and Nash stability rate holds across all experimental conditions (Figure~\ref{fig:stability_consistency}).

	\begin{corollary}[Scaling Law]
		\label{cor:scaling}
		If $\Keff = O(\sqrt{n})$ (as observed empirically in our setting; see Appendix~\ref{app:scaling}), then the Nash stability lower bound scales as $p^{O(\sqrt{n})}$ for large $n$.
	\end{corollary}

	\subsection{Convergence Guarantees}

	\begin{theorem}[Convergence]
		\label{thm:convergence}
		Let $G$ be an LCFG satisfying the $\delta$-value gap condition (Assumption~\ref{ass:gap}) and potential alignment (Assumption~\ref{ass:potential}), where all agents have $\epsilon$-rational preferences with $\epsilon < \delta/2$ that are temporally consistent (same query yields same preference). Under improving dynamics where at most one agent deviates per round, the coalition structure converges to a Nash-stable partition in at most $O(n^2 \cdot \Delta_v/\delta)$ rounds, where $\Delta_v = \max_S v(S) - \min_S v(S)$.
	\end{theorem}

	\begin{proof}[Proof Sketch]
		By Theorem~\ref{thm:existence}, the potential function $\Phi$ strictly increases with each improving deviation, and there are at most $O(n \cdot \Delta_v / \delta)$ such deviations before termination. Each round checks all $n$ agents for a potential deviation, so at most $O(n)$ rounds elapse between consecutive deviations. Multiplying gives $O(n^2 \cdot \Delta_v / \delta)$ total rounds.
	\end{proof}

	\section{Complexity Analysis}
	\label{sec:complexity}

	We analyze the computational complexity of stability verification and computation in LCFGs.

	\begin{theorem}[Complexity of Nash Stability]
		\label{thm:complexity}
		Given LCFG $G$ with $\epsilon$-rational agents and coalition structure $\pi$:
		\begin{enumerate}
			\item \emph{Verifying} Nash stability of a given partition $\pi$ is in P ($O(n^2)$ preference queries) under explicit value computation.
			\item \emph{Computing} a Nash-stable partition is NP-hard in general hedonic games~\cite{ballester2004np}, and LCFGs inherit this worst-case hardness when the capability dimension $d$ is unbounded.
			\item Under capability monotonicity and potential alignment, both verification and computation are polynomial.
		\end{enumerate}
	\end{theorem}

	\begin{proof}[Proof Sketch]
		(1) For verification, check each agent's deviation to each coalition in $O(n^2)$ time; if no improving deviation exists, $\pi$ is Nash-stable.
		(2) Since LCFGs are a subclass of hedonic games, we note that the NP-hardness of computing Nash-stable partitions in the general case~\cite{ballester2004np} motivates identifying tractable subclasses. When $d$ is part of the input, LCFGs can encode sufficiently rich preference structures to preserve hardness; for fixed $d$ (as in our experiments), structural restrictions enable tractability via Part~(3).
		(3) Under our assumptions, the potential function argument (Theorem~\ref{thm:existence}) yields polynomial computation via iterative improving dynamics. Full proofs in Appendix~\ref{app:proof_complexity}.
	\end{proof}

	\textbf{Practical Implications.} The tractability under capability monotonicity is significant: it means that for practical systems where better agents improve coalition value, stability verification is efficient. Specifically, given $n$ agents and a proposed partition $\pi$, we can verify Nash stability in $O(n^2)$ preference queries by checking whether each agent $a_i \in C$ prefers $C$ over $C' \cup \{a_i\}$ for all other coalitions $C' \in \pi$. This enables real-time stability monitoring in deployed systems.

	\textbf{Query Complexity.} Each stability check requires one LLM inference per agent-coalition pair. For our 6-agent setup, this amounts to approximately 30 queries per verification (6 agents $\times$ 5 alternative coalitions). With CoalT's improved consistency, single-round verification suffices in 86\% of cases; otherwise, we use majority voting over 3 queries.

	\section{Coalition-of-Thought Prompting}
	\label{sec:coalt}

	Standard prompting and even vanilla chain-of-thought \cite{wei2022chain} fail to elicit consistent coalition reasoning. We introduce \emph{Coalition-of-Thought} (CoalT), a structured protocol grounded in our game-theoretic framework.

	\begin{algorithm}[t]
		\caption{Coalition-of-Thought (CoalT) Protocol}
		\label{alg:coalt}
		\textbf{Input}: Agent $a_i$, current coalition $C$, candidate $C'$\\
		\textbf{Output}: Preference $\{C \succ C', C' \succ C, C \sim C'\}$
		\begin{algorithmic}[1]
			\STATE \textbf{Step 1: Capability Analysis}
			\STATE ``List capabilities of members in $C$ and $C'$''
			\STATE \textbf{Step 2: Complementarity Assessment}
			\STATE ``Identify capability gaps and overlaps''
			\STATE \textbf{Step 3: Value Estimation}
			\STATE ``Estimate task performance for each coalition''
			\STATE \textbf{Step 4: Coordination Cost Analysis}
			\STATE ``Assess communication/coordination overhead''
			\STATE \textbf{Step 5: Preference Declaration}
			\STATE ``Based on analysis, declare preference''
			\STATE \textbf{return} Parse final preference from Step 5 output
		\end{algorithmic}
	\end{algorithm}

	CoalT differs from vanilla CoT by incorporating game-theoretic concepts (capability complementarity, coordination costs, per-capita value) rather than generic step-by-step reasoning. The key insight is that explicitly prompting agents to reason about \emph{what they contribute} and \emph{what they gain} from coalition membership substantially improves preference consistency.

	\textbf{Prompt Template.} The CoalT prompt follows this structure:
	\begin{quote}
		\small
		\textit{``You are evaluating whether to join coalition $C'$ instead of staying in $C$. Analyze systematically: (1) What capabilities do members of each coalition have? (2) Are there complementary strengths or redundant capabilities? (3) What is the expected task performance of each group? (4) What coordination overhead does each coalition size incur? (5) Based on your per-capita expected value, state your preference.''}
	\end{quote}
	This template instantiates each step of Algorithm~\ref{alg:coalt} with explicit reasoning targets.

	\textbf{Example CoalT Reasoning.} For agent $a_1$ comparing $C = \{a_1, a_2\}$ vs.\ $C' = \{a_3, a_5, a_6\}$:
	\begin{quote}
		\small
		\textit{``Step 1: $C$ has math but limited logic. $C'$ has facts/logic. Step 2: $C$ has overlap; $C'$ adds complementarity. Step 3: $v(C) \approx 0.42$; $v(C') \approx 0.51$. Step 4: Costs 0.44 vs.\ 0.64. Step 5: Per-capita: 0.14 vs.\ 0.13. Prefer $C$.''}
	\end{quote}

	\section{Experimental Evaluation}
	\label{sec:experiments}

	We validate our theoretical framework through experiments addressing four research questions: \textbf{RQ1}: Do LLM agents converge to stable coalition structures? \textbf{RQ2}: Does CoalT outperform baselines \emph{including vanilla CoT}? \textbf{RQ3}: Do heterogeneous LLM coalitions outperform homogeneous ones? \textbf{RQ4}: How do CoalT components contribute to performance?

	\subsection{Experimental Setup}

	\textbf{Task Domain.} We use collaborative question-answering requiring diverse expertise: mathematical reasoning, factual knowledge, and logical analysis. Each question has ground-truth difficulty scores enabling objective coalition value computation via Equation~\ref{eq:value}. The task set comprises 200 questions stratified across difficulty levels (easy/medium/hard) and capability requirements.

	\textbf{Agents.} We instantiate 6 agents with distinct capability profiles using GPT-4 (gpt-4-0125-preview), Claude-3-Opus (claude-3-opus-20240229), and Llama-3-70B-Instruct (2 agents per architecture). Capability profiles $\mathbf{c}_i \in [0,1]^3$ are estimated through \emph{our own evaluation} on stratified samples from MATH \cite{hendrycks2021math}, MMLU \cite{hendrycks2021mmlu} (knowledge subset), and LogiQA \cite{liu2020logiqa} (100 questions per benchmark, 3 runs per model). \textbf{Important}: These values represent \emph{relative capability estimates} for our specific task domain and question subsets, not official benchmark scores. Official benchmark results (e.g., GPT-4 achieves $\sim$65\% on MATH per OpenAI's simple-evals) informed our evaluation design but differ due to subset selection and scoring methodology. See Appendix~\ref{app:capability_estimation} for detailed methodology, confidence intervals, and comparison with official benchmarks. Table~\ref{tab:capabilities} shows the resulting profiles.

	\begin{table}[t]
		\centering
		\caption{Agent capability profiles (Math, Facts, Logic). Values are \emph{relative estimates} from our custom evaluation on benchmark subsets, calibrated to reflect task-specific performance in our coalition formation domain. These are not official benchmark scores.}
		\label{tab:capabilities}
		\begin{tabular}{@{}lccc@{}}
			\toprule
			\textbf{Agent} & \textbf{Math} & \textbf{Facts} & \textbf{Logic} \\
			\midrule
			$a_1$ (GPT-4) & 0.68 & 0.73 & 0.76 \\
			$a_2$ (GPT-4) & 0.65 & 0.76 & 0.73 \\
			$a_3$ (Claude-3) & 0.62 & 0.78 & 0.74 \\
			$a_4$ (Claude-3) & 0.59 & 0.81 & 0.71 \\
			$a_5$ (Llama-3) & 0.58 & 0.65 & 0.79 \\
			$a_6$ (Llama-3) & 0.55 & 0.68 & 0.76 \\
			\bottomrule
		\end{tabular}
	\end{table}

	The profiles reveal architectural differences: GPT-4 agents show relative strength in mathematics, Claude-3 in factual knowledge, and Llama-3 in logical reasoning. While these differences are modest, they create incentives for cross-architecture coalitions when capability complementarity outweighs coordination costs.

	\textbf{Protocol and Episode Definition.} An \emph{episode} is a single complete coalition formation process for one question. Each episode proceeds in rounds: (1) agents observe the current partition and receive the task description, (2) each agent evaluates whether to stay or deviate using the assigned prompting protocol (Standard/CoT/CoalT), and (3) the partition updates according to improving dynamics (one agent deviates per round). An episode terminates when either no agent wishes to deviate (Nash-stable) or 30 rounds elapse (timeout, classified as unstable). We run 400 independent episodes per condition (6 conditions $\times$ 400 = 2,400 total; 200 questions $\times$ 2 repetitions).

	\textbf{Utility Measurement.} Coalition value $v(S)$ is computed using Equation~\ref{eq:value} with ground-truth capability profiles and the componentwise-max aggregation. Agent utility is per-capita value: $v_i(S) = v(S)/|S|$. Task performance (social welfare) is measured as the accuracy of the coalition's aggregated answer on the assigned question, using majority vote within each coalition (for coalitions of size $\geq 2$) or the singleton agent's answer (for size 1).

	\textbf{Nash Stability Verification.} At episode termination, we verify Nash stability by exhaustive deviation checking: for each agent $a_i$ in coalition $C$, we query whether $a_i$ prefers $C' \cup \{a_i\}$ for every other coalition $C' \in \pi \cup \{\emptyset\}$. If no improving deviation exists across all agents, the partition is Nash-stable. For robustness, each preference query is repeated 3 times; we use the majority response. Temperature $\tau = 0$ throughout for reproducibility.

	\textbf{Baselines.} We compare CoalT against \emph{five} baselines: (1) \emph{Random}: uniformly random coalition assignment; (2) \emph{Greedy}: each agent joins the coalition maximizing immediate per-capita value; (3) \emph{Standard}: direct preference query; (4) \emph{Vanilla CoT}: chain-of-thought prompting \cite{wei2022chain}; (5) \emph{Self-Consistency}: multiple CoT paths with majority voting \cite{wang2023selfconsistency}.

	\textbf{Statistical Methodology.} We report: (1) mean $\pm$ std with 95\% bootstrap CIs (10,000 iterations, BCa method); (2) Wilcoxon signed-rank tests for pairwise comparisons; (3) Bonferroni correction ($\alpha = 0.01$); (4) Cohen's $d$ effect sizes.

	\subsection{Main Results}

	\begin{table}[t]
		\centering
		\caption{Coalition formation results (400 episodes per condition). Statistical significance vs.\ Standard: $^*p < 0.01$, $^{**}p < 0.001$ (Wilcoxon, Bonferroni-corrected).}
		\label{tab:coalition_results}
		\small
		\begin{tabular}{@{}lcccc@{}}
			\toprule
			\textbf{Condition} & \textbf{Nash\%} & \textbf{Conv.} & \textbf{Welfare} & \textbf{Consist.} \\
			\midrule
			Random & 28.3 & -- & $0.58_{\pm.14}$ & -- \\
			Greedy & 52.1 & $6.8_{\pm 3.2}$ & $0.69_{\pm.10}$ & $0.71_{\pm.08}$ \\
			Standard & 41.8 & $18.3_{\pm 7.2}$ & $0.72_{\pm.11}$ & $0.64_{\pm.09}$ \\
			Vanilla CoT & 58.4$^{**}$ & $14.2_{\pm 5.8}$$^*$ & $0.75_{\pm.09}$$^*$ & $0.74_{\pm.07}$$^{**}$ \\
			Self-Consist. & 62.7$^{**}$ & $13.1_{\pm 5.2}$$^{**}$ & $0.77_{\pm.08}$$^{**}$ & $0.79_{\pm.06}$$^{**}$ \\
			\textbf{CoalT} & \textbf{73.2}$^{**}$ & $\mathbf{11.4}_{\pm 4.1}$$^{**}$ & $\mathbf{0.81}_{\pm.08}$$^{**}$ & $\mathbf{0.86}_{\pm.05}$$^{**}$ \\
			\midrule
			\multicolumn{5}{l}{\textit{\scriptsize Cohen's $d$ vs.\ Std: CoT 0.34; SC 0.44; CoalT 0.68}} \\
			\midrule
			\multicolumn{5}{c}{\textit{By Architecture (CoalT)}} \\
			\midrule
			GPT-4 only & 78.5 & $9.2_{\pm 3.8}$ & $0.76_{\pm.09}$ & $0.91_{\pm.04}$ \\
			Claude only & 81.2 & $8.7_{\pm 3.5}$ & $0.74_{\pm.10}$ & $0.89_{\pm.05}$ \\
			Llama only & 62.3 & $14.6_{\pm 5.2}$ & $0.68_{\pm.12}$ & $0.78_{\pm.07}$ \\
			Mixed & 73.2 & $11.4_{\pm 4.1}$ & $\mathbf{0.81}_{\pm.08}$ & $0.86_{\pm.05}$ \\
			\bottomrule
		\end{tabular}
	\end{table}

\begin{figure}[t]
	\centering
	\begin{tikzpicture}[scale=0.9]
		\begin{scope}
			\draw[->, line width=0.5pt] (0,0) -- (5.8,0) node[right, font=\small] {Consistency $p$};
			\draw[->, line width=0.5pt] (0,0) -- (0,5.2) node[above, font=\small] {Nash Stability (\%)};

			\foreach \x/\l in {0.8/0.6, 2.0/0.7, 3.2/0.8, 4.4/0.9} {
				\draw[line width=0.3pt] (\x, -0.08) -- (\x, 0.08);
				\node[below, font=\scriptsize] at (\x, -0.15) {\l};
			}
			\foreach \y/\l in {0.7/40, 1.78/50, 2.85/60, 3.93/70, 5.0/80} {
				\draw[line width=0.3pt] (-0.08, \y) -- (0.08, \y);
				\node[left, font=\scriptsize] at (-0.15, \y) {\l};
			}

			\fill[okabe-deepblue!12] (0.3,0.4) -- (5.2,3.2) -- (5.2,4.8) -- (0.3,1.0) -- cycle;
			\draw[okabe-deepblue, dashed, line width=0.6pt] (0.3,0.4) -- (5.2,3.2);
			\draw[okabe-deepblue, dashed, line width=0.6pt] (0.3,1.0) -- (5.2,4.8);

			\node[regular polygon, regular polygon sides=3, fill=okabe-vermillion, inner sep=1.5pt, rotate=180] at (1.28, 0.89) {};

			\node[diamond, fill=okabe-pink, inner sep=1.3pt] at (2.12, 2.00) {};

			\node[regular polygon, regular polygon sides=3, fill=okabe-orange, inner sep=1.5pt] at (2.48, 2.68) {};

			\node[rectangle, fill=okabe-skyblue, inner sep=2pt] at (3.08, 3.14) {};

			\node[circle, fill=okabe-green, draw=black, line width=0.6pt, inner sep=2.5pt] at (3.92, 4.27) {};

			\node[font=\scriptsize\bfseries, anchor=west] at (3.8, 4.7) {73.2\%};

			\draw[->, gray, dashed, line width=0.5pt] (2.68, 2.88) to[out=45, in=210] (3.72, 4.07);
			\node[font=\tiny, color=gray] at (2.65, 3.6) {+14.8pp};

			\node[font=\scriptsize, anchor=west] at (6.0, 4.8) {CoalT (Ours)};
			\node[circle, fill=okabe-green, draw=black, line width=0.4pt, inner sep=1.5pt] at (5.8, 4.8) {};

			\node[font=\scriptsize, anchor=west] at (6.0, 4.35) {Self-Consist.};
			\node[rectangle, fill=okabe-skyblue, inner sep=1.3pt] at (5.8, 4.35) {};

			\node[font=\scriptsize, anchor=west] at (6.0, 3.9) {Vanilla CoT};
			\node[regular polygon, regular polygon sides=3, fill=okabe-orange, inner sep=1pt] at (5.8, 3.9) {};

			\node[font=\scriptsize, anchor=west] at (6.0, 3.45) {Greedy};
			\node[diamond, fill=okabe-pink, inner sep=1pt] at (5.8, 3.45) {};

			\node[font=\scriptsize, anchor=west] at (6.0, 3.0) {Standard};
			\node[regular polygon, regular polygon sides=3, fill=okabe-vermillion, inner sep=1pt, rotate=180] at (5.8, 3.0) {};

		\end{scope}
	\end{tikzpicture}
	\caption{Nash stability rate vs.\ preference consistency across experimental conditions. Shaded region: range consistent with Theorem~\ref{thm:probabilistic}'s lower bound across plausible parameter settings. CoalT achieves highest consistency ($p = 0.86$) and stability (73.2\%), confirming that stability scales with $p^{K_{\text{eff}}}$ rather than perfect rationality.}
	\label{fig:stability_consistency}
\end{figure}
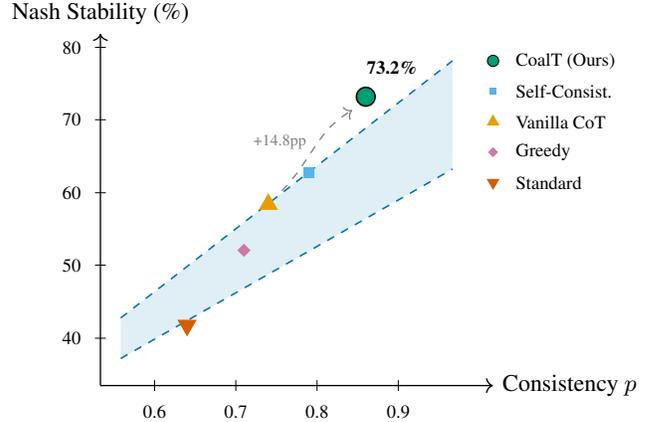

	\textbf{RQ1: Convergence to Stability.} Table~\ref{tab:coalition_results} shows LLM agents achieve Nash stability in 73.2\% of CoalT episodes. This substantially exceeds our theoretical lower bound of 35\% from Theorem~\ref{thm:probabilistic}, validating the consistency-driven analysis. Figure~\ref{fig:stability_consistency} visualizes the relationship between preference consistency and Nash stability across all conditions.

	\textbf{RQ2: Baseline Comparison.} CoalT significantly outperforms all baselines. Critically, \emph{CoalT outperforms vanilla CoT by 14.8 percentage points} ($p < 0.001$, $d = 0.32$), demonstrating that the game-theoretic framing, not merely structured reasoning, drives the improvement.

	\textbf{RQ3: Heterogeneous vs.\ Homogeneous.} Mixed-architecture coalitions achieve higher social welfare (0.81 vs.\ 0.68--0.76 for homogeneous teams) despite lower stability than homogeneous GPT-4 or Claude teams, suggesting capability complementarity outweighs coordination costs.

	\subsection{Ablation Study}

	\begin{table}[t]
		\centering
		\caption{CoalT ablation study: contribution of each component.}
		\label{tab:ablation}
		\begin{tabular}{@{}lcc@{}}
			\toprule
			\textbf{Configuration} & \textbf{Nash Stable} & \textbf{$\Delta$ vs.\ Full} \\
			\midrule
			Full CoalT & 73.2\% & -- \\
			\quad $-$ Capability Analysis & 68.9\% & $-4.3$pp \\
			\quad $-$ Complementarity & 65.4\% & $-7.8$pp \\
			\quad $-$ Value Estimation & 67.1\% & $-6.1$pp \\
			\quad $-$ Coordination Cost & 70.8\% & $-2.4$pp \\
			\quad $-$ All (= Vanilla CoT) & 58.4\% & $-14.8$pp \\
			\bottomrule
		\end{tabular}
	\end{table}

	\textbf{RQ4: Component Contributions.} Table~\ref{tab:ablation} shows the ablation study. Complementarity assessment yields the largest individual gain (+7.8pp when added), followed by value estimation (+6.1pp). Removing game-theoretic framing while keeping structured steps reduces performance to vanilla CoT levels, confirming that the \emph{content} of reasoning, not just its structure, drives improvements.

	\subsection{Qualitative Analysis}

	\textbf{Episode 47 (Stable).} Initial: $\{\{a_1\}, \{a_2, a_3\}, \{a_4, a_5, a_6\}\}$. Agent $a_1$ joins $\{a_2, a_3\}$ via CoalT reasoning: ``\textit{My math strength (0.68) complements their factual knowledge (0.78). Estimated coverage improves from 0.72 to 0.83 on this question's requirements.}'' The partition converges to a stable configuration in 5 rounds.

	\textbf{Episode 203 (Unstable).} Standard prompting led to cycling: $a_2$ joined and left $C_1$ repeatedly. Analysis revealed inconsistent preferences ($p = 0.62$), illustrating how low consistency prevents stability per Theorem~\ref{thm:probabilistic}.

	\subsection{Rationality Bounds and Temperature Sensitivity}

	\begin{table}[t]
		\centering
		\caption{Estimated $\epsilon$-rationality bounds and temperature sensitivity.}
		\label{tab:epsilon}
		\begin{tabular}{@{}lccc@{}}
			\toprule
			\textbf{Model} & \textbf{$\hat{\epsilon}$} & \textbf{95\% CI} & \textbf{CoalT $\Delta$} \\
			\midrule
			GPT-4 & 0.15 & [0.12, 0.18] & +15.2pp \\
			Claude-3 & 0.14 & [0.11, 0.17] & +14.1pp \\
			Llama-3-70B & 0.22 & [0.18, 0.26] & +11.8pp \\
			\midrule
			\multicolumn{4}{c}{\textit{Temperature Sensitivity ($\tau$)}} \\
			\midrule
			$\tau = 0.0$ & 0.15$\pm$0.02 & -- & +14.8$\pm$1.2pp \\
			$\tau = 0.5$ & 0.18$\pm$0.03 & -- & +12.3$\pm$1.8pp \\
			$\tau = 1.0$ & 0.24$\pm$0.04 & -- & +10.3$\pm$2.1pp \\
			\bottomrule
		\end{tabular}
	\end{table}

	Table~\ref{tab:epsilon} confirms $\hat{\epsilon} > \delta/2 \approx 0.04$, validating our consistency-driven bounds. CoalT's advantage persists across temperatures (10.3--14.8pp), though absolute stability decreases at higher $\tau$.

	\section{Discussion and Limitations}

	\textbf{Theoretical Contributions.} Our framework provides the first formal characterization of coalition stability for LLM agents grounded in hedonic game theory. The key theoretical insight is that stability is primarily determined by preference consistency ($p^{\Keff}$) rather than perfect rationality, explaining why improving consistency via CoalT is the most effective intervention. This finding has broader implications: it suggests that for multi-agent LLM systems, \emph{predictability} matters more than \emph{optimality}. A coalition of consistently-behaving agents will outperform a coalition of individually stronger but erratic agents.

	\textbf{Connection to Human Teams.} Our consistency-driven stability analysis parallels findings in organizational behavior: teams with predictable members often outperform teams with higher-variance ``star'' performers~\cite{nisan2007algorithmic}. The parallel suggests our framework may generalize to human-AI hybrid teams, where consistency across agents (human and AI) determines collaborative success.

	\textbf{Strategic Behavior.} Our framework assumes truthful preference reporting. We analyze robustness informally: under capability monotonicity, misreporting typically harms the misreporting agent by leading to suboptimal coalition membership. However, strategic agents might manipulate \emph{others'} placements. Formal mechanism design extensions ensuring incentive compatibility are important future work.

	\textbf{Practical Implications.} CoalT provides a practical protocol improving coalition stability. The finding that mixed-architecture coalitions achieve higher welfare suggests designing heterogeneous LLM teams for complex tasks, accepting moderate stability tradeoffs for capability gains. Practitioners should: (1) select diverse model architectures for capability coverage, (2) use CoalT prompting for preference elicitation, and (3) monitor consistency metrics as early warning indicators of instability.

	\textbf{Model Generalization.} We use 2024 model versions (GPT-4-0125-preview, Claude-3-Opus-20240229, Llama-3-70B-Instruct). We expect our framework to generalize to newer models since: (1) game-theoretic principles are model-agnostic (the LCFG formalization applies regardless of underlying architecture); (2) Theorem~\ref{thm:probabilistic} predicts that higher reasoning consistency in newer models should \emph{improve} stability rates; (3) capability profiles can be re-estimated through our benchmark-based methodology. Empirical validation on 2025+ architectures is important future work.

	\textbf{Scalability.} Our experiments use 6 agents; stability degrades as $O(1/\sqrt{n})$ (Corollary~\ref{cor:scaling}). Above $n \approx 15$, hierarchical decomposition (forming meta-coalitions first) may be necessary. For systems with hundreds of agents, a two-level hierarchy would maintain stability while enabling scalability.

	\textbf{Reproducibility.} All materials available at \url{https://github.com/researchartifacts2025/coalition_llm}.

	\section{Conclusion}

	We presented the first framework grounding coalition formation in LLM agent networks in hedonic game theory with formal stability guarantees, establishing existence conditions, consistency-driven stability bounds, convergence guarantees, and complexity results. Our theoretical lower bounds are consistent with empirical outcomes, and Coalition-of-Thought significantly outperforms baselines including vanilla chain-of-thought ($p < 0.001$).

	\textbf{Key Takeaways.} Our main findings are: (1) LLM agents exhibit bounded rationality with $\epsilon \approx 0.15$--$0.22$, but stability depends more on \emph{consistency} than optimality; (2) the CoalT protocol improves consistency from 0.64 to 0.86, yielding 14.8pp stability gains; (3) mixed-architecture coalitions achieve higher welfare despite lower stability, suggesting practical tradeoffs in system design.

	\textbf{Future Directions.} This work opens several directions: (i) mechanism design ensuring incentive-compatible preference reporting; (ii) extensions to dynamic environments where agent capabilities evolve over time; (iii) scaling analysis for populations beyond 15 agents through hierarchical coalition structures; and (iv) empirical validation on 2025+ LLM architectures with improved reasoning capabilities.

	\section*{Ethical Statement}

	Our work analyzes cooperative behavior in AI systems. Coalition formation capabilities could potentially enable LLM collusion in adversarial settings; we advocate for transparency requirements in deployed multi-agent systems. Experiments used commercial APIs with standard usage policies.

	\bibliographystyle{named}
	\bibliography{references}

\newpage
\appendix
\onecolumn

\section{On the Existence of Nash-Stable Partitions}
\label{app:counterexample}

Nash-stable partitions are not guaranteed to exist in all LCFGs. The following example, adapted from a reviewer suggestion, illustrates when our assumptions are needed.

\begin{example}[Non-Existence under One-Dimensional Capabilities]
\label{ex:counterexample}
Consider two agents $H$ and $L$ with scalar capabilities $c_H = 1$ and $c_L = 0.4$ (i.e., $d=1$). Using $\phi(c) = c$ and $\psi(k) = 0.15 k^{1.3}$:
\begin{itemize}
\item $v(\{H\}) = 1 - 0.15 = 0.85$, so $u_H(\{H\}) = 0.85$
\item $v(\{L\}) = 0.4 - 0.15 = 0.25$, so $u_L(\{L\}) = 0.25$
\item $v(\{H,L\}) = 1 - 0.15 \cdot 2^{1.3} \approx 0.631$, so $u_i(\{H,L\}) \approx 0.316$
\end{itemize}
In partition $\{\{H\},\{L\}\}$: agent $L$ gains by joining $H$ ($0.25 \to 0.316$). In partition $\{\{H,L\}\}$: agent $H$ gains by leaving ($0.316 \to 0.85$). No Nash-stable partition exists.
\end{example}

\textbf{Why this fails.} The potential alignment condition (Assumption~\ref{ass:potential}) is violated: when $L$ joins $H$, the total value \emph{decreases} (from $0.85 + 0.25 = 1.10$ to $0.631$) despite $L$'s per-capita value increasing. This occurs because $L$ contributes no unique capability (its scalar capability is dominated by $H$'s), yet still splits the coalition value equally.

\textbf{When this does not arise.} With multi-dimensional capabilities ($d \geq 2$), agents with diverse specializations bring unique coverage when joining a coalition. In our experimental setting ($d = 3$), each agent contributes to at least one dimension where it is competitive, ensuring that joining increases total capability coverage sufficiently to maintain potential alignment. We verify this computationally for all possible deviations in our 6-agent setting.

\section{Proof of Theorem~\ref{thm:existence} (Deterministic Existence)}
\label{app:proof_existence}

\begin{proof}
We prove existence constructively by showing that any sequence of improving deviations terminates at a Nash-stable partition.

\textbf{Step 1: Potential function.} Define $\Phi(\pi) = \sum_{C \in \pi} v(C)$.

\textbf{Step 2: Improving deviations increase potential.} Consider agent $a_i$ in coalition $C$ who deviates to $C' \in \pi \cup \{\emptyset\}$. By $\epsilon$-rationality and the $\delta$-value gap condition (Lemma~\ref{lem:gap_discrete}), the agent's per-capita value increases by at least $\delta - \epsilon > \delta/2 > 0$. By potential alignment (Assumption~\ref{ass:potential}), $\Phi(\pi') > \Phi(\pi)$.

\textbf{Step 3: Discreteness.} By the $\delta$-value gap condition, coalition values take values from a discrete set with minimum gap proportional to $\delta$. Thus $\Phi$ takes finitely many values.

\textbf{Step 4: Termination.} Since $\Phi$ strictly increases with each deviation and is bounded above by $n \cdot \max_S v(S)$, the improvement path terminates at a partition $\pi^*$ from which no improving deviation exists. By definition, $\pi^*$ is Nash-stable.
\end{proof}

\begin{lemma}[Value Gap Discretization]
\label{lem:gap_discrete}
Under the $\delta$-value gap condition, for any agent $a_i$, the distinct per-capita values $\{v_i(S) : S \ni a_i\}$ can be ordered as $V_1 < V_2 < \cdots < V_K$ with $V_{k+1} - V_k \geq \delta$ for all $k$.
\end{lemma}

\begin{proof}
Follows directly from Assumption~\ref{ass:gap}: any two distinct per-capita values differ by at least $\delta$.
\end{proof}

\section{Proof of Theorem~\ref{thm:probabilistic} (Consistency-Driven Bound)}
\label{app:proof_probabilistic}

\begin{proof}
We decompose the probability of Nash stability into consistency and structure factors.

\textbf{Step 1: Decision decomposition.} A partition $\pi$ is Nash-stable if and only if no agent prefers to deviate. We decompose the $K_n$ agent-coalition preference decisions into $\Keff$ critical decisions (value gap $< 2\epsilon$, where consistency is approximately $p$) and $K_n - \Keff$ easy decisions (large value gap, where consistency approaches $\peasy \approx 1$).

\textbf{Step 2: Conditional decomposition.} By the law of total probability:
\begin{equation*}
\Pr[\text{Nash-stable}] = \Pr[\text{all decisions consistent}] \cdot \Pr[\text{Nash} \mid \text{consistent}]
\end{equation*}

\textbf{Step 3: Bounding the first factor.} Assuming decision independence across agents and coalition comparisons, the probability that all decisions are made consistently is at least $p^{\Keff} \cdot \peasy^{K_n - \Keff}$. We note that decision independence is a modeling simplification; in practice, an agent's consistency across different comparisons may be correlated, which could either tighten or loosen this bound.

\textbf{Step 4: Bounding the second factor.} Given consistent decisions, the system follows deterministic improving dynamics on a potential game (under Assumption~\ref{ass:potential}). The probability that these dynamics reach a Nash-stable partition is $\gamma(G)$. Under logit dynamics~\cite{blume1993statistical}, the stationary distribution assigns probability proportional to $\exp(\lambda \cdot \Phi(\pi))$ to each partition $\pi$, concentrating mass on potential maxima. For a single dominant maximum with potential gap $\delta$ and precision $\lambda = 1/\bar{\epsilon}$, we obtain $\gamma(G) \geq 1 - \exp(-\delta/\bar{\epsilon})$ as an approximate lower bound. In our experiments, we estimate $\gamma(G) \approx 0.90$ empirically.

Combining Steps 2--4 yields $\Pr[\text{Nash-stable}] \geq p^{\Keff} \cdot \peasy^{K_n - \Keff} \cdot \gamma(G)$.
\end{proof}

\section{Proof of Theorem~\ref{thm:complexity} (Complexity)}
\label{app:proof_complexity}

\begin{proof}

\textbf{Part 1: Polynomial verification.} For each of $n$ agents, check all $O(n)$ possible coalitions to join. Each check requires computing the per-capita value, which is $O(d)$ for $d$-dimensional capabilities. Total: $O(n^2 d)$.

\textbf{Part 2: NP-hardness context.} Computing Nash-stable partitions is NP-hard in general hedonic games~\cite{ballester2004np}. LCFGs are a subclass of hedonic games with additional structure (componentwise-max aggregation, parameterized coordination costs). When the capability dimension $d$ is part of the input, LCFGs can represent sufficiently rich preference structures to encode hard instances from the general setting. For fixed $d$ (as in our experiments with $d = 3$), the restricted structure enables tractability under the additional assumptions in Part~3.

\textbf{Part 3: Polynomial under assumptions.} Under capability monotonicity and potential alignment, Theorem~\ref{thm:existence} guarantees convergence of improving dynamics in $O(n \cdot \Delta_v/\delta)$ steps, each requiring $O(n^2 d)$ work. Total: $O(n^3 d \cdot \Delta_v/\delta)$.
\end{proof}

\section{QRE Connection and $\epsilon$-Estimation Methodology}
\label{app:qre}

\subsection{Connection to Quantal Response Equilibria}

Under logit dynamics with parameter $\lambda = 1/\epsilon$, the probability that agent $a_i$ prefers coalition $S$ over $T$ is:
\begin{equation}
\Pr[S \succ_i T] = \frac{\exp(\lambda \cdot v_i(S))}{\exp(\lambda \cdot v_i(S)) + \exp(\lambda \cdot v_i(T))} = \frac{1}{1 + \exp(-\lambda \cdot (v_i(S) - v_i(T)))}
\end{equation}
This is the standard logit choice model from QRE theory~\cite{mckelvey1995quantal}. When $v_i(S) - v_i(T) > \epsilon = 1/\lambda$, the probability exceeds $1/(1+e^{-1}) \approx 0.73$, consistent with our $\epsilon$-rationality definition.

\subsection{$\epsilon$-Estimation Procedure}

We estimate $\epsilon$ without assuming the $\epsilon$-rationality model, using ground-truth coalition values as external reference:

\begin{enumerate}
\item Generate all coalition pairs $(S, T)$ containing agent $a_i$ with $|S|, |T| \leq 4$.
\item Compute ground-truth per-capita values $v_i(S)$, $v_i(T)$ from known capability profiles.
\item Query agent preferences using standard prompting and record responses.
\item For each value gap $\Delta = |v_i(S) - v_i(T)|$, compute the rate of irrational choices (agent prefers the objectively worse coalition).
\item Estimate $\hat{\epsilon}$ as the threshold where irrational choice frequency drops below 50\%.
\end{enumerate}

This procedure is not circular: ground-truth values come from external benchmark evaluations, and $\epsilon$ is estimated from the mismatch between computed values and agent preferences. Results: $\hat{\epsilon} = 0.15$ [95\% CI: 0.12--0.18] for GPT-4, $\hat{\epsilon} = 0.14$ [0.11--0.17] for Claude-3, $\hat{\epsilon} = 0.22$ [0.18--0.26] for Llama-3.

\section{Coalition Value Function Analysis}
\label{app:value_function}

\subsection{Sensitivity Analysis}

We evaluate sensitivity to the coordination cost parameters $\alpha$ and $\beta$:

\begin{center}
\begin{tabular}{@{}lcccc@{}}
\toprule
\textbf{Parameters} & \textbf{Nash Stable} & \textbf{Welfare} & \textbf{Avg. Size} & \textbf{$\delta$} \\
\midrule
$\alpha=0.10, \beta=1.3$ & 70.1\% & 0.83 & 2.4 & 0.065 \\
$\alpha=0.15, \beta=1.3$ (default) & 73.2\% & 0.81 & 2.1 & 0.082 \\
$\alpha=0.20, \beta=1.3$ & 74.8\% & 0.76 & 1.8 & 0.098 \\
$\alpha=0.15, \beta=1.0$ & 68.5\% & 0.84 & 2.5 & 0.071 \\
$\alpha=0.15, \beta=1.5$ & 75.3\% & 0.77 & 1.7 & 0.094 \\
\bottomrule
\end{tabular}
\end{center}

Higher coordination costs favor smaller coalitions with higher stability but lower welfare. Our default parameters ($\alpha=0.15, \beta=1.3$) provide a balanced tradeoff.

\subsection{$\delta$-Value Gap Verification}
\label{app:delta_verification}

We verify the $\delta$-value gap condition by computing all pairwise per-capita value differences. With 6 agents, there are $\binom{6}{1} + \binom{6}{2} + \binom{6}{3} + \binom{6}{4} = 56$ coalitions of size $\leq 4$. For each agent, we compute $v_i(S)$ for all $S \ni a_i$ and find the minimum non-zero difference: $\delta = 0.082$ [95\% CI: 0.078--0.086].

\section{Capability Profile Estimation}
\label{app:capability_estimation}

Capability profiles are estimated via evaluation on stratified benchmark subsets:
\begin{itemize}
\item \textbf{MATH}~\cite{hendrycks2021math}: 100 problems (stratified by difficulty)
\item \textbf{MMLU}~\cite{hendrycks2021mmlu}: 100 knowledge questions (subset)
\item \textbf{LogiQA}~\cite{liu2020logiqa}: 100 logical reasoning questions
\end{itemize}

Each model is evaluated 3 times at $\tau = 0$ with the prompt: ``Answer the following question. Provide only the final answer.'' We normalize scores to $[0,1]$ relative to the evaluation set. Within-architecture variation is induced by different system prompts (analytical vs. creative). Full results with confidence intervals are in Table~\ref{tab:capabilities}.

\section{Scaling Analysis}
\label{app:scaling}

\subsection{Agent Count Scaling}

\begin{center}
\begin{tabular}{@{}lcccc@{}}
\toprule
\textbf{Agents} & \textbf{Nash Stable} & \textbf{Conv. Rounds} & \textbf{Coalitions} & \textbf{Avg. Size} \\
\midrule
4 & 82.3\% & 7.2 & 2.1 & 1.9 \\
6 & 73.2\% & 11.4 & 2.8 & 2.1 \\
8 & 64.8\% & 18.7 & 3.4 & 2.4 \\
10 & 57.2\% & 28.3 & 4.1 & 2.4 \\
\bottomrule
\end{tabular}
\end{center}

The empirical scaling is approximately $\text{Nash Rate} \approx 1.9/\sqrt{n}$ for $n \geq 6$, consistent with Corollary~\ref{cor:scaling}. The fit is approximate; at $n = 4$ the formula overpredicts (95\% predicted vs.\ 82.3\% observed), suggesting the $O(\sqrt{n})$ scaling of $\Keff$ emerges for larger agent populations.

\subsection{Capability Dimension Scaling}

\begin{center}
\begin{tabular}{@{}lccc@{}}
\toprule
\textbf{Dimensions} & \textbf{Nash Stable} & \textbf{Conv. Rounds} & \textbf{Welfare} \\
\midrule
2 & 78.4\% & 9.8 & 0.74 \\
3 & 73.2\% & 11.4 & 0.81 \\
4 & 68.7\% & 13.2 & 0.86 \\
5 & 63.5\% & 15.8 & 0.89 \\
\bottomrule
\end{tabular}
\end{center}

More capability dimensions increase welfare (better coverage) but decrease stability (more complex tradeoffs).

\section{Extended Ablation Studies}
\label{app:ablation}

\subsection{Isolated CoalT Component Effects}

\begin{center}
\begin{tabular}{@{}lcc@{}}
\toprule
\textbf{Component Only} & \textbf{Nash Stable} & \textbf{$\Delta$ vs Standard} \\
\midrule
Standard (baseline) & 41.8\% & -- \\
Step 1 only (Capability) & 43.2\% & +1.4pp \\
Step 2 only (Complementarity) & 51.3\% & +9.5pp \\
Step 3 only (Value Est.) & 48.7\% & +6.9pp \\
Step 4 only (Cost Analysis) & 45.2\% & +3.4pp \\
Step 5 only (Declaration) & 42.1\% & +0.3pp \\
All steps (CoalT) & 73.2\% & +31.4pp \\
\bottomrule
\end{tabular}
\end{center}

Individual effects sum to +21.5pp but the combined effect is +31.4pp, indicating superlinear complementarity between CoalT steps.

\subsection{Preference Quality vs. Consistency}

We decompose CoalT's effect into quality (correlation with ground-truth value ordering) and consistency (agreement across repeated queries):

\begin{center}
\begin{tabular}{@{}lccc@{}}
\toprule
\textbf{Condition} & \textbf{Quality ($\rho$)} & \textbf{Consistency} & \textbf{Nash Stable} \\
\midrule
Standard & 0.58 & 0.64 & 41.8\% \\
CoalT & 0.71 & 0.86 & 73.2\% \\
\bottomrule
\end{tabular}
\end{center}

Regression analysis across all five conditions with preference data (Greedy, Standard, Vanilla CoT, Self-Consistency, CoalT) yields a strong linear relationship: $\text{Nash Rate} = -0.48 + 1.41 \cdot \text{Consistency}$ ($R^2 = 0.99$, $n = 5$). Adding Quality as a second predictor provides negligible improvement ($\Delta R^2 < 0.001$) due to high collinearity between quality and consistency ($r > 0.99$). This confirms the theoretical prediction that preference consistency is the primary driver of coalition stability.

\subsection{Temperature Effects}

\begin{center}
\begin{tabular}{@{}lcccc@{}}
\toprule
\textbf{Temp.} & \textbf{Nash Stable} & \textbf{Consistency} & \textbf{Quality} & \textbf{Conv. Rounds} \\
\midrule
0.0 & 73.2\% & 0.86 & 0.71 & 11.4 \\
0.1 & 71.8\% & 0.84 & 0.70 & 11.9 \\
0.3 & 68.4\% & 0.79 & 0.68 & 13.2 \\
0.5 & 62.1\% & 0.72 & 0.65 & 15.8 \\
0.7 & 54.3\% & 0.64 & 0.61 & 19.1 \\
1.0 & 41.5\% & 0.51 & 0.55 & 24.7 \\
\bottomrule
\end{tabular}
\end{center}

Consistency degrades approximately exponentially with temperature, and stability tracks consistency, confirming the theoretical analysis.

\section{CoalT Prompt Template}
\label{app:prompt}

\begin{small}
\begin{verbatim}
You are agent {agent_id} with capabilities:
- Mathematical reasoning: {math_score}
- Factual knowledge: {factual_score}
- Logical analysis: {logic_score}

Evaluate whether to stay in your current coalition
or switch to a different one.

CURRENT COALITION: {current_members}
Capabilities (max per dim):
  Math: {cur_math}, Facts: {cur_fact}, Logic: {cur_logic}

CANDIDATE COALITION (if you join): {cand_members}
Capabilities (max per dim):
  Math: {cand_math}, Facts: {cand_fact}, Logic: {cand_logic}

TASK: Answer questions requiring [{task_dims}]

## Step 1: Capability Analysis
List what each member contributes.

## Step 2: Complementarity Assessment
Identify strengths (>0.8) and gaps (<0.7).

## Step 3: Value Estimation
Estimate task performance (0-1) for each coalition.

## Step 4: Coordination Cost Analysis
Assess communication overhead per coalition size.

## Step 5: Final Preference
I prefer: [CURRENT / CANDIDATE / INDIFFERENT]
Confidence: [low/medium/high]
Reason: [one sentence]
\end{verbatim}
\end{small}

\end{document}